\documentclass[sigconf,natbib=true]{acmart}
\AtBeginDocument{%
  \providecommand\BibTeX{{%
    \normalfont B\kern-0.5em{\scshape i\kern-0.25em b}\kern-0.8em\TeX}}}

\copyrightyear{2024}
\acmYear{2024}
\setcopyright{acmlicensed}\acmConference[SIGIR '24]{Proceedings of the 47th International ACM SIGIR Conference on Research and Development in Information Retrieval}{July 14--18, 2024}{Washington, DC, USA} \acmBooktitle{Proceedings of the 47th International ACM SIGIR Conference on Research and Development in Information Retrieval (SIGIR '24), July 14--18, 2024, Washington, DC, USA}
\acmDOI{10.1145/3626772.3657720} 
\acmISBN{979-8-4007-0431-4/24/07}


\usepackage{algorithm}
\usepackage{algorithmicx}
\usepackage{algpseudocode}

\usepackage{booktabs}
\usepackage{multirow}
\usepackage{subfigure}
\usepackage{makecell}
\usepackage{tikz}
\usepackage{pgfplots}
\usepackage{pgfplotstable}
\usepackage{enumitem}

\begin{document}

\title{GPT4Rec: Graph Prompt Tuning for Streaming Recommendation}

\author{Peiyan Zhang}
\authornote{Both authors contributed equally to this research.}
\affiliation{\institution{Hong Kong University of \\ Science and Technology}\country{Hong Kong}}
\email{pzhangao@cse.ust.hk}

\author{Yuchen Yan}
\authornotemark[1]
\affiliation{\institution{School of Intelligence Science and Technology, Peking University}\city{Beijing}\country{China}}
\email{2001213110@stu.pku.edu.cn}

\author{Xi Zhang}
\affiliation{\institution{Interdisciplinary Institute for Medical Engineering, Fuzhou University}\city{Fuzhou}\country{China}}
\email{zxwinner@gmail.com}

\author{Liying Kang}
\affiliation{\institution{Hong Kong Polytechnic University}\country{Hong Kong}}
\email{lykangc12@gmail.com}

\author{Chaozhuo Li}
\authornote{Chaozhuo Li is the corresponding author}
\affiliation{\institution{Microsoft Research Asia}\city{Beijing}\country{China}}
\email{lichaozhuo1991@gmail.com}

\author{Feiran Huang}
\affiliation{\institution{Jinan University}\city{Guangzhou}\country{China}}
\email{huangfr@jnu.edu.cn}

\author{Senzhang Wang}
\affiliation{\institution{Central South University}\city{Changsha}\country{China}}
\email{szwang@csu.edu.cn}

\author{Sunghun Kim}
\affiliation{\institution{Hong Kong University of \\ Science and Technology}\country{Hong Kong}}
\email{hunkim@cse.ust.hk}

\renewcommand{\shortauthors}{Peiyan Zhang et al.}

\begin{abstract} 

In the realm of personalized recommender systems, the challenge of adapting to evolving user preferences and the continuous influx of new users and items is paramount. Conventional models, typically reliant on a static training-test approach, struggle to keep pace with these dynamic demands. Streaming recommendation, particularly through continual graph learning, has emerged as a novel solution, attracting significant attention in academia and industry. However, existing methods in this area either rely on historical data replay, which is increasingly impractical due to stringent data privacy regulations; or are inability to effectively address the over-stability issue; or depend on model-isolation and expansion strategies, which necessitate extensive model expansion and are hampered by time-consuming updates due to large parameter sets. To tackle these difficulties, we present GPT4Rec, a \textbf{G}raph \textbf{P}rompt \textbf{T}uning method for streaming \textbf{Rec}ommendation. Given the evolving user-item interaction graph, GPT4Rec first disentangles the graph patterns into multiple views. After isolating specific interaction patterns and relationships in different views, GPT4Rec utilizes lightweight graph prompts to efficiently guide the model across varying interaction patterns within the user-item graph. Firstly, node-level prompts are employed to instruct the model to adapt to changes in the attributes or properties of individual nodes within the graph. Secondly, structure-level prompts guide the model in adapting to broader patterns of connectivity and relationships within the graph. Finally, view-level prompts are innovatively designed to facilitate the aggregation of information from multiple disentangled views. These prompt designs allow GPT4Rec to synthesize a comprehensive understanding of the graph, ensuring that all vital aspects of the user-item interactions are considered and effectively integrated. Experiments on four diverse real-world datasets demonstrate the effectiveness and efficiency of our proposal.

\end{abstract}



\maketitle

\section{Introduction}
\par 
Recommender Systems (RSs) have become indispensable in shaping personalized experiences across a multitude of domains, profoundly influencing user choices in e-commerce, online streaming, web searches, and so forth~\citep{zhang2023efficiently,guo2022evolutionary, zhou2023exploring,jin2022code,liu2023chatgpt}. RSs not only guide users through an overwhelming array of options but also drive engagement and customer satisfaction, making them critical to the success of digital platforms. Among the diverse techniques employed to decode complex user preferences, Graph Neural Networks (GNNs)~\citep{Wu2019SessionbasedRW,Pan2020StarGN,chen2020handling} stand out as a groundbreaking approach. GNNs adeptly unravel the intricate patterns of user-item interactions, significantly enhancing the precision and effectiveness of recommendations~\citep{li2019adversarial,zhang2023can,zhao2022learning}. 

\par However, these methods deployed in the real world often underdeliver its promises made through the benchmark datasets~\citep{zhang2023survey}. This discrepancy largely stems from their traditional offline training and testing approach~\citep{mi2020ader}. In these scenarios, models are trained on large, static datasets and then evaluated on limited test sets, a process that doesn't account for the dynamic nature of real-world data. In stark contrast, real-world RS are in a state of constant flux, where new user preferences, items, and interactions continually emerge, creates a gap that is essentially the difference in data distributions over time. On one hand, the models that are originally trained on historical data might not be well-equipped to handle such new, diverse data effectively. On the other hand, when these models are updated with the new data, they are at risk of overwriting the knowledge previously acquired—a phenomenon known as \textit{Catastrophic Forgetting}~\citep{chang2017streaming}. 
This issue is notably problematic in RS, where retaining older but pertinent information is pivotal for sustaining a holistic grasp of user preferences and behavior. Consequently, although GNN-based RS models demonstrate considerable prowess, their ability to adapt to the perpetually changing data landscape poses a significant challenge requiring urgent and concentrated efforts.

\par Recent studies have aimed to embrace this challenge, and most of these works are delving into harnessing the potential of continual learning methods~\citep{chen2013terec,diaz2012real,mi2020memory,mi2020ader,qiu2020gag}.
The first line of research~\citep{chaudhry2019continual,isele2018selective,mi2020ader,prabhu2020gdumb,rebuffi2017icarl} relies on a replay buffer to periodically retrain the model using a selection of past samples. However, the effectiveness of such sample-based methods diminishes with a reduced buffer size and becomes impractical in scenarios where using a replay buffer is constrained, such as in situations requiring strict data privacy~\citep{milano2020recommender}. This limitation is crucial: when the buffer fails to represent the complete spectrum of past data, the method struggles to preserve essential historical knowledge, leading to a gap between what is retained and the current data landscape.

The second line of works, model regularization-based methods~\citep{dhar2019learning,hinton2015distilling,hou2019learning,rannen2017encoder}, aims to maintain knowledge by constraining the model’s parameters to prevent significant divergence from previously learned configurations. These parameters are critical as they often encapsulate patterns extracted from historical data. Yet, the challenge arises when new data diverges substantially from these historical patterns. If the model’s parameters are not adequately adaptable to this new information, it risks straying too far from relevant past data, triggering catastrophic forgetting. The last line of works relies on the model isolation and expansion strategies~\citep{golkar2019continual,ostapenko2021continual,qin2021bns}. These strategies isolate old knowledge and create new learning spaces for updated data. Yet, their extensive model expansion often results in increased parameters and time-consuming updates. In essence, while these strategies appear promising, they fail to fully satisfy the key requirement of streaming recommendation: effectively bridging the ever-present gap between evolving new data and past data distributions. It's not just about preventing catastrophic forgetting but also about ensuring effective learning and adaptation to new data. This dual requirement is where these methods fall short, underscoring the need for more sophisticated strategies that can seamlessly integrate evolving data dynamics while retaining essential historical insights, thereby addressing the core challenge of streaming recommendation.

\par In order to tackle the aforementioned challenges, we draw inspiration from the concept of prompt tuning~\citep{liu2023pre}, a new transfer learning technique in the field of natural language processing (NLP). Intuitively, prompt tuning techniques reformulates learning downstream tasks from directly adapting model weights to designing prompts that “instruct” the model to perform tasks conditionally. A prompt encodes task-specific knowledge and has the ability to utilize pre-trained frozen models more effectively than ordinary fine-tuning~\citep{lester2021power,raffel2020exploring,wang2022learning}. This effectiveness stems from the prompts' ability to add contextual layers to the model's understanding, thereby adapting its responses to new data without altering the core model structure. Moreover, prompt tuning stands as a data-agnostic technique. Unlike methods that heavily rely on the data they were trained on, prompt tuning can navigate across different data distributions without being hindered by the gaps typically encountered in evolving datasets. This quality makes it immune to the pitfalls introduced by the gap in data distributions. Thus, prompt tuning emerges as a potent solution for meeting the dual requirements of continual learning in RS. It not only aids in preventing catastrophic forgetting by maintaining the integrity of the model's foundational knowledge but also ensures effective learning and adaptation to new and diverse data patterns.

In particular, applying prompt tuning in continual graph learning for recommender systems is challenging. First of all, existing prompt tuning approaches~\citep{smith2023coda,wang2022learning,razdaibiedina2023progressive,yang2022semantic,hu2023pop} are for data in Euclidean space, e.g., images, texts. In the dynamic user-item interaction graphs, however, changes are not just incremental; they are cascaded and interconnected, profoundly influencing the entire network. For example, the addition or removal of a node triggers a domino effect, altering the states of adjacent nodes and potentially leading to substantial changes across the entire graph. This cascading nature of change implies that incremental updates simultaneously affect multiple levels of relationships within the graph. Therefore, how to disentangle the multiple levels of relationship changes caused by the cascaded changing of graph data becomes a unique challenge for graph prompt tuning. Furthermore, once these relationships are disentangled, a significant task remains in re-aggregating these representations while preserving the overall graph's integrity. Efficiently combining these decoupled elements is vital for maintaining a coherent and accurate representation of the evolving graph. Last but not least, a major concern with existing methods is the lack of theoretical guarantees for the prompt tuning process in dynamic environments. This limitation can lead to unstable and inconsistent training, exacerbated by the continually changing nature of graph data in recommender systems.

To tackle the challenges above, in this paper, we introduce GPT4Rec, a \textbf{G}raph \textbf{P}rompt \textbf{T}uning method for streaming \textbf{Rec}ommendation. To address the challenge of managing cascaded changes, GPT4Rec disentangles the graph into distinct feature and structure views, which is optimized to capture unique characteristics of each semantic relationship within the graph. Central to GPT4Rec's design are three types of prompts tailored for specific aspects of graph data. Firstly, node-level prompts are employed to instruct the model to adapt to changes in the attributes or properties of individual nodes within the graph, ensuring a nuanced understanding of evolving node characteristics. Secondly, structure-level prompts guide the model in adapting to broader patterns of connectivity and relationships within the graph, capturing the dynamic interplay between different graph elements. Finally, view-level prompts are innovatively designed to facilitate the aggregation of information from multiple disentangled views. This approach allows GPT4Rec to synthesize a comprehensive understanding of the graph, ensuring that all vital aspects of the user-item interactions are considered and effectively integrated. The utilization of lightweight graph prompts efficiently guides the model across varying interaction patterns within the user-item graph. This approach is in stark contrast to model-isolation methods, as GPT4Rec's prompt based strategy enables rapid adaptation to new data streams. It efficiently pinpoints and transfers useful components of prior knowledge to new data, thereby preventing the erasure of valuable insights.
Finally, we provide theoretical analysis specifically from the perspective of graph data to justify the ability of our method. We theoretically show that GPT4Rec has at least the expression ability of fine-tuning globally using the whole data. Extensive experiments are conducted on six real-world datasets, where GPT4Rec outperforms state-of-the-art baselines for continual learning on dynamic graphs.

In summary, we make the following contributions:
\begin{itemize}
    \item We propose GPT4Rec, a graph prompt tuning based approach tailored for streaming recommendation. By strategically utilizing node-level, structure-level, and view-level prompts, GPT4Rec effectively guides the model to recognize and incorporate new data trends without overwriting valuable historical knowledge. 
    \item Theoretical analyses affirm that GPT4Rec has at least the expression ability of fine-tuning globally.
    \item We conduct extensive evaluations on four real-world datasets, where GPT4Rec achieves state-of-the-art on streaming recommendation.
\end{itemize}
\section{Related Work}
\subsection{Streaming Recommendation}
Traditional recommender systems, constrained by static datasets, struggle with predicting shifting user preferences and trends due to the dynamic nature of user interactions and the expanding volume of items. Streaming recommendation, a dynamic approach updating both data and models over time, addresses these challenges~\citep{chang2017streaming,chen2013terec,das2007google,devooght2015dynamic,song2017multi,song2008real,wang2018streaming}. While initial efforts focused on item popularity, recency, and trend analysis~\citep{chandramouli2011streamrec,lommatzsch2015real,subbian2016recommendations}, recent advancements integrate collaborative filtering and matrix factorization into streaming contexts~\citep{chang2017streaming,devooght2015dynamic,diaz2012real,rendle2008online}. Further, approaches using online clustering of bandits and collaborative filtering bandits have emerged~\citep{ban2021local,gentile2017context,gentile2014online, li2019improved,li2016collaborative}. The application of graph neural networks (GNNs) in streaming recommendation models is gaining attention for its complex relationship modeling~\citep{ahrabian2021structure,wang2020streaming,wang2022streaming,zhang2023continual}. This shift towards streaming recommendation systems represents a significant advancement in the field, offering a more dynamic and responsive approach to user preference analysis and item suggestion.

\subsection{Continual Learning}
Continual learning (CL) addresses the sequential task processing with strategies to prevent catastrophic forgetting and enable knowledge transfer. The primary algorithms in continual learning are categorized into three groups: experience replay~\citep{chaudhry2019continual,isele2018selective,mi2020ader,prabhu2020gdumb,rebuffi2017icarl}, model regularization~\citep{dhar2019learning,hinton2015distilling,hou2019learning,rannen2017encoder}, and model isolation~\citep{golkar2019continual,ostapenko2021continual,qin2021bns}. Recently, continual graph learning~\citep{zhang2023continual,cai2022multimodal,liu2021overcoming,ma2020streaming,pareja2020evolvegcn,perini2022learning} has emerged, focusing on chronological data in streaming recommendation systems~\citep{ahrabian2021structure,wang2020streaming,wang2021graph,xu2020graphsail}. the focus shifts to handling data that arrives continuously in a chronological sequence, rather than data segmented by tasks. This research diverges from traditional continual learning by emphasizing effective knowledge transfer across time segments, instead of solely focusing on preventing catastrophic forgetting.

\subsection{Graph Prompt tuning}
Prompt tuning, a technique extensively employed in Natural Language Processing (NLP) and Computer Vision (CV), aims to bridge the gap between pre-training tasks and fine-tuning tasks. This methodology has recently seen increased application in graph-based scenarios, highlighting its versatility and effectiveness~\citep{yan2024inductive}. In the realm of graph neural networks (GNNs), several innovative adaptations of prompt tuning have emerged: GPPT~\citep{sun2022gppt} leverages learnable graph label prompts, transforming the node classification task into a link prediction task to mitigate the task type gap. GraphPrompt~\citep{liu2023graphprompt} introduces a universal prompt template to unify all the tasks via a learnable readout function. All-in-One~\citep{sun2023all} proposes an inventive graph token prompt, coupled with a token insertion strategy to align the pre-training and fine-tuning tasks. HGPrompt~\citep{Yu2023HGPROMPTBH} extends the concept of prompt tuning to heterogeneous graphs by designing unique prompts for each node type. In the context of streaming recommendations, our work marks a pioneering effort in utilizing graph prompt tuning. We employ prompts to guide the model in swiftly adapting to new data streams, ensuring the continuous integration of evolving data while maintaining the integrity of previously learned information.
\section{PRELIMINARIES}
In this section, we first formalize the continual graph learning for streaming recommendation. Then we briefly introduce three classical graph convolution based recommendation models used in this paper.

\subsection{Definitions and Formulations}

\textbf{Definition 1. Streaming Recommendation.} Industrial recommender systems handle a continuous influx of user-item interaction data, denoted as $\tilde{D}$. This stream is divided into sequential segments $D_{1},...,D_{t},...,D_{T}$, each covering an equal time period~\citep{wang2022streaming,zhang2023continual}. In each time segment $t$, the model optimizes recommendation performance on $D_{t}$ by incorporating knowledge from previous segments $D_{1},D_{2},...,D_{t-1}$. The performance is evaluated over the entire timeline.


\noindent\textbf{Definition 2. Streaming Graph.} A streaming graph is defined as a series of graphs $G=(G_{1},G_{2},...,G_{t},...G_{T})$, where each graph $G_{t}$ evolves from its predecessor by incorporating changes $\Delta G_{t}$ such that $G_{t} = G_{t-1} + \Delta G_{t}$. At time $t$, the attributed graph $G_{t} = (A_{t}, X_{t})$ consists of an adjacency matrix $A_{t}$ and node feature matrix $X_{t}$. The incremental changes $\Delta G_{t} = (\Delta A_{t}, \Delta X_{t})$ reflect modifications in graph structure and node attributes, including the addition/deletion of nodes and edges.


\noindent\textbf{Definition 3. Continual Graph Learning (CGL) for Streaming Graph.} Given a streaming graph $G=(G_{1},G_{2},...,G_{t},...G_{T})$, the objective of CGL is to sequentially learn the updates $\Delta G_{t}(D_{t})$ while effectively transferring knowledge from historical data to new graph segments. Formally, CGL aims to learn the optimal graph structure $S_{t}$ and parameters $W_{t}$ for each segment $t$, formulated as:
\begin{equation}
    (S^{*}_{t},W^{*}_{t})=\mathop{\arg\min}\limits_{S_{t},W_{t}}\mathcal{L}_{t}(S_{t},W_{t},\Delta G_{t}),
\end{equation}
where $S_{t}$ and $W_{t}$ are elements of the search spaces $\mathcal{S}$ and $\mathcal{W}$, respectively. The function $\mathcal{L}{t}(S{t}, W_{t}, \Delta G_{t})$ denotes the loss for the current segment, evaluated on $\Delta G_{t}$.

\section{Methodology}
In this section, we introduce our GPT4Rec method towards continual graph learning for streaming recommendation. As shown in Figure~\ref{fig:overview}, we first disentangle the complex user-item interaction graphs into multiple views, capturing specific interaction aspects. We then design node-level prompts to focus on individual node attributes and structure-level prompts to address broader connectivity patterns. Afterwards, we propose the cross-view-level prompts to aggregate information from these views, integrating new knowledge efficiently and adaptively.

\begin{figure*}[t]
\centering
\includegraphics[width=0.80\textwidth]{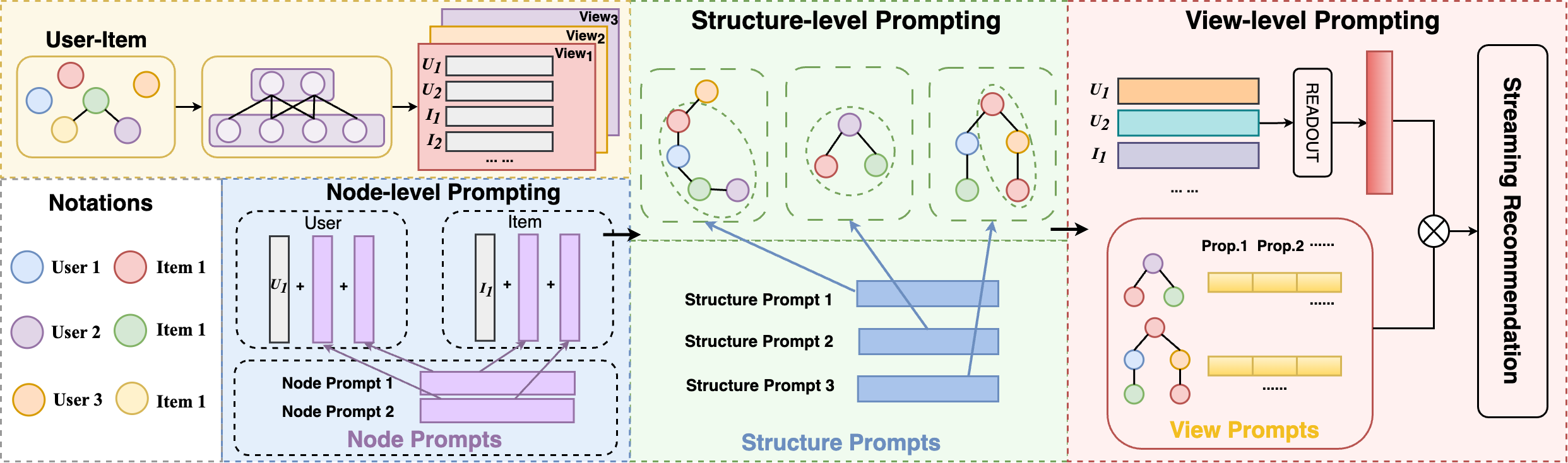}
\caption{Overview of the GPT4Rec. a) We disentangle the complex user-item interaction graph into multiple views. b) We design the node-level prompt to depict the node-level changes. c) We propose the structure-level prompt for the decomposition of sub-graph structure. d) Finally, we present a view-level prompt to aggregate the disentangled views.}
\label{fig:overview}
\end{figure*}

\subsection{Disentangling Strategy for Complex Graphs}

The user-item interaction graphs in recommender systems are inherently complex due to their dynamic and interconnected nature. When a new node is added or an existing one is removed, it doesn't just affect isolated parts of the graph. Instead, these changes can lead to cascading effects throughout the network. These cascading changes mean that any modification in the graph can simultaneously influence multiple relationships. For example, a new item added to the system could cultivate a new preference for several users, altering the existing user-item interaction patterns and potentially reshaping the overall landscape of preferences and recommendations. Such dynamics underscore the complexity of these graphs and the need for a modeling approach that can adapt to and capture these multifaceted and simultaneous changes.

In this context, we divide the graph into multiple views, where each view is tailored to capture specific aspects of the user-item interactions. The disentanglement is achieved through a series of linear transformations:
\begin{equation}
    \tilde{x}_{i}={Linear}_{i}(x),
\end{equation}
where ${Linear}_{i}$ is the linear transformation of the $i$-th view, $i\in\{1,2,...,N\}$. The linear transformations are capable of separating these views while maintaining the overall integrity of the graph. This means the model can isolate and focus on specific aspects without losing sight of the graph's interconnected nature.
This disentanglement allows the model to explore and identify distinct patterns and relationships in the data. For example, one view might capture user-to-item interactions, another might focus on user-to-user relationships, and yet another could delve into item-to-item similarities. With each view focusing on a particular aspect of the graph, the model can adapt more precisely to changes or updates in the data.

\subsection{Prompt Design for Adaptive Learning}
After disentangling the graph patterns into different views, we design node-level prompt and structure-level prompt to capture the comprehensive essence of graph patterns in dynamic RSs. 

\subsubsection{Node-level Prompts.} The node-level prompts primarily target the attributes or properties of individual nodes within the graph. This could include user characteristics in a social network or item properties in a recommendation system. By focusing on this level, GPT4Rec can delve into the intricacies of node-specific data, allowing for a nuanced understanding of individual behaviors or item features. This focus is crucial for tasks where personalization or detailed attribute analysis is key.

Specifically, for each view, the node-level prompts is a set of learnable parameters $P = [p_{1},...,p_{L}]$, where $L$ is the number of node-level prompts. These prompts act as targeted cues that inform the model about how to interpret and integrate new information regarding users or items.

\noindent\textbf{Contextual Guidance Through Weighted Addition.} When new data arrives, these prompts effectively 'instruct' the model by highlighting relevant features or changes in the user-item interactions:
\begin{equation}
    \tilde{x}_{i} = \tilde{x}_{i}+ \Sigma_{j}^{L} \alpha_{ij}p_{j},
\end{equation}
\begin{equation}
    \alpha_{ij}=\frac{exp(\tilde{x}_{i}p_{j})}{\Sigma_{r}exp(\tilde{x}_{i}p_{r})},
\end{equation}
where $\alpha_{ij}$ is the weight for prompt $j$ calculated based on its relevance to the current data point $\tilde{x}_{i}$. This is done using the softmax function, which essentially turns raw scores (obtained by multiplying $x_{i}$ and $p_{j}$) into probabilities. 

These prompts, each encoding specific patterns or relationships, are weighed differently for different nodes. This means that for a given node $x_{i}$, certain prompts will have higher weights ($\alpha_{ij}$) if they are more relevant to that node’s context. This selective amplification allows the model to focus on aspects of the data that are currently most pertinent. For instance, if a particular prompt encodes a pattern that is increasingly common in new data, its weight will be higher for nodes where this pattern is relevant. 

As new data comes in, the relevance of different prompts can change. The model dynamically recalculates the weights ($\alpha_{ij}$) for each node with each new data point, allowing it to adapt its focus continuously. This dynamic process ensures that the model remains responsive to evolving data trends and relationships, integrating new knowledge in a way that's informed by both the new and existing data.

\subsubsection{Structure-level Prompts.}
Alongside the node-level prompts, structure-level prompts are designed to engage with the broader patterns of connectivity and relationship within the graph. These prompts are crucial for understanding and adapting to changes in the overall graph topology, such as the emergence of new interaction patterns or the evolution of existing ones.

The structure-level prompt is designed as follows:  
for each view, we design a set of learnable prompt $Q=[q_{1},...,q_{K}]$ for the edges that adaptively aggregate the structure-level information via message-passing mechanism:
\begin{equation}
    \tilde{x}_{i}=\tilde{x}_{i}+\Sigma_{j\in N(x_{i})}\beta_{ij}u_{ij}, 
\end{equation}
\begin{equation}
    \beta_{ij}=\frac{exp(\tilde{x}_{i}u_{ij})}{\Sigma_{r}\tilde{x}_{i}u_{ir}},
\end{equation}
\begin{equation}
    u_{ij}=\Sigma_{k}^{K}att_{ij,k}q_{k},
\end{equation}
\begin{equation}
att_{ij,k}=\frac{exp({Linear}(\tilde{x}_{i}||\tilde{x}_{j})q_{k})}{\Sigma_{r}exp({Linear}(\tilde{x}_{i}||\tilde{x}_{j})q_{r})},
\end{equation}
where $x_{i}$ and $x_{j}$ are adjacent nodes.

The integration of these prompts within each view facilitates a comprehensive and responsive learning process. By simultaneously addressing both the granular details at the node level and the broader structural dynamics, GPT4Rec ensures a holistic understanding of the graph data. This approach enables the model to effectively learn from and adapt to the continually evolving landscape of user-item interactions in streaming recommendation.

\subsection{Aggregation of Disentangled Representations}
The aggregation of information from multiple disentangled views is crucial for providing a comprehensive understanding of the dynamic and interconnected user-item interactions. 

\subsubsection{Initial strategy.}

One straightforward approach to this aggregation is the application of an attention mechanism. For the disentangled views, the aggregation can be formulated as:
\begin{equation}
    \hat{x}_{i}=Atten({Linear}(x_{i}), [\tilde{x}_{i,1},...,\tilde{x}_{i,n}])=x_{i} + \Sigma_{j}\gamma_{i,j}\tilde{x}_{i,j},
\end{equation}
\begin{equation}
    \gamma_{i,j}=\frac{exp({Linear}(x_{i})\tilde{x}_{i,j})}{\Sigma_{r}exp({Linear}(x_{i})\tilde{x}_{i,r})},
\end{equation}
However, this approach may not fully account for the evolving nature of user-item interactions, especially with the introduction of new data. As the graph's structure changes, the relevance of different views and their interrelations can shift, potentially rendering the existing fixed attention weights less effective.

\subsubsection{Cross-view-level prompts for aggregation.}

To enhance the model's efficiency and adaptability in the face of these dynamic changes, GPT4Rec incorporates Cross-View-Level Prompts for dynamic adaptation. This approach centers on updating a small set of `codebook' prompts $Z=[z_{1},...,z_{N}]$, rather than relearning the entire model's parameters. These prompts serve as dynamic modifiers to the attention mechanism, allowing the model to adapt its focus efficiently:
\begin{equation}
    \hat{x}_{i}=att(prompt(x_{i}), [\tilde{x}_{i,1},...,\tilde{x}_{i,N}])=x_{i} + \Sigma_{j}^{N}\epsilon_{(i,z,j)}\tilde{x}_{i,j},
\end{equation}
\begin{equation}
    \epsilon{(i,z,j)}=\frac{exp(prompt(x_{i})\tilde{x}_{i,j})}{\Sigma_{r}exp(prompt(x_{i})\tilde{x}_{i,r})} = \frac{exp({Linear}(x_{i}+z_{i})\tilde{x}_{i,j})}{\Sigma_{r}exp({Linear}(x_{i}+z_{i})\tilde{x}_{i,r})},
\end{equation}

In this enhanced aggregation process, the prompts subtly adjust the attention weights, reflecting the current state and relationships within the graph. This strategy maintains the model's stability while enabling it to respond dynamically to new data, ensuring the final node embedding $\hat{x}$ remains relevant and accurate over time.

\subsection{Discussions}
In this section, we discuss the differences between graph prompt tuning, as employed in GPT4Rec, and traditional model-isolation-expansion methods~\citep{golkar2019continual,ostapenko2021continual,qin2021bns}, particularly addressing the nature of knowledge storage and adaptation in these approaches.differences between graph prompt tuning, as employed in GPT4Rec, and traditional model-isolation-expansion methods, particularly addressing your queries about the nature of knowledge storage and adaptation in these approaches. The core distinction lies in how each approach integrates new knowledge and preserves existing information.

Traditional model-isolation-expansion methods typically involve expanding the model's capacity to accommodate new information. This often means adding new layers or nodes, effectively increasing the model's size and complexity. While this approach can be effective in integrating new knowledge, it often requires significant resources and can lead to a bloated model. The expansion needs to be substantial enough to capture the new information, which might not be efficient or scalable in the long term. These methods can sometimes struggle with the delicate balance of preserving existing knowledge while incorporating new data. The expansion can dilute the model's original understanding, potentially leading to issues like overfitting to recent data at the expense of older, yet still relevant, insights.

On the contrary, graph prompt tuning doesn't simply expand the model's space to store new knowledge. Instead, it introduces a set of contextually adaptive prompts that act as conduits for the new information. These prompts do not store knowledge in the traditional sense; they modify how the model interprets and processes incoming data. The prompts in GPT4Rec provide nuanced guidance, subtly adjusting the model's focus and understanding based on the current context. They act as dynamic, lightweight 'instructors' that align new data with the model's existing knowledge base. This is achieved without the need for extensive expansion of the model's structure, ensuring efficiency and agility. The key here is adaptability. The prompts are designed to be flexible, adjusting their influence based on the relevance to new data. This allows GPT4Rec to seamlessly integrate new insights while maintaining the integrity of previously learned information, thus avoiding catastrophic forgetting.

\subsection{Theoretical Analysis}
We conduct theoretical analysis to guarantee the correctness of the proposed  graph prompt tuning algorithm on dynamic graphs. The conclusion is achieved by the following theorem with its proofs.
\begin{theorem}
GPT4Rec has at least the expression ability of fine-tuning globally using the whole data.
\end{theorem}
\begin{proof}
Suppose the model is updated at the time $t$, which means the model parameter ${\theta}_{t}$ is optimal. At time $t+1$, the global fine-tuning process is as follows (taking $x_{i}$ as the example): 
\begin{equation}
    \mathop{\arg\min}\limits_{\theta_{t+1}} L(f_{\theta_{t+1}}(x_{i}^{t+1}), y_{i}^{t+1}),
    \label{eq:orignal}
\end{equation}
where we use $\theta_{t}$ to initialize $\theta_{t+1}$ and $L$ is the loss function.  The optimization has the upper bound as\cite{zhang2023continual}:

\begin{equation}
    \mathop{\arg\min}\limits_{\theta^{t+1}} L(f_{\theta_{t+1}}(\Delta x_{i}^{t}), y_{i}^{t+1})+ L(f_{\theta_{t}}(x_{i}^{t}),y_{i}^{t}),
    \label{eq:diff}
\end{equation}
where $\Delta x_{i}^{t}$ denotes the node data gap, Equation~\ref{eq:diff} equals to Equation~\ref{eq:orignal} when $\theta_{t+1}=0$ and, representing the initialization process.

For the optimization of prompt, the optimization process is:
\begin{equation}
\begin{aligned}
    & \mathop{\arg\min}\limits_{{promp}^{t+1}}L(f_{\theta_{t}}(x_{i}^{t+1}+\Sigma_{j}\epsilon_{i,j}{promp}(x_{i}^{t+1})), y_{i}^{t+1}), \\
    \Leftrightarrow & \mathop{\arg\min}\limits_{{promp}^{t+1}}L(f_{\theta_{t}}(\Delta x_{i}^{t+1}+\Sigma_{j}\epsilon_{i,j}{promp^{t+1}}(x_{i}^{t+1})+,y_{i}^{t+1}) + \\
    & L(f_{\theta^{t}}(x_{i}^{t}, y_{i}^{t+1}), \\
\end{aligned}
\end{equation}
which is equivalent to we fix $\theta_{t}$ and optimize Equation~\ref{eq:orignal}. Namely, we optimize the gap from the optimal.
Therefore, GPT4Rec has at least the expression ability of global fine-tuning with the whole data.
\end{proof}


\section{Experiments}
In this section, we conduct experiments on four real-world time-stamped recommendation datasets to evaluate our proposal. 

\subsection{Experiment Settings}
\label{sec:expe_setup}
\subsubsection{Datasets.} We conduct experiments on four datasets (\textit{i.e.,} Neflix\footnote{\noindent https://academictorrents.com/details/9b13183dc4d60676b773c9e2cd6de5e5542cee9a}, Foursquare\footnote{\noindent https://sites.google.com/site/yangdingqi/home/foursquare-dataset}~\citep{yang2019revisiting,yang2020lbsn2vec++} and Taobao2014\footnote{\noindent https://tianchi.aliyun.com/dataset/46} and Taobao2015\footnote{\noindent https://tianchi.aliyun.com/dataset/dataDetail?dataId=53} from from Alibaba’s M-Commerce platforms) from three different domains (i.e., social media, points-of-interests , and e-commerce). Following~\cite{he2023dynamically,yang2019revisiting,yang2020lbsn2vec++}, we use the average entity overlapping rate (AER) between segments to assess data stream stability; higher AER indicates greater stability. Data in each segment is divided into training, validation, and test sets in an 8:1:1 ratio. The statistics of the four datasets are summarized in Table~\ref{tab:dt_descrip}.

\begin{table}[t]
\centering
  \caption{Statistics of datasets used in the experiments.}
  \label{tab:dt_descrip}
  \resizebox{0.48\textwidth}{!}{
  \begin{tabular}{c c c c c}
    \toprule
    Statistic & \textit{Tb2014} & \textit{Tb2015} & \textit{Netflix} & \textit{Foursquare}\\
    \midrule
    No. of Users & 8K & 192K & 301K & 52K \\
    No. of Items & 39K & 10K & 9K & 37K  \\
    No. of Interactions  & 749K & 9M & 49M & 2M  \\
    Time Span & 31 \text{days} & 123 \text{days} & 74 \text{months} & 22 \text{months}  \\
    AER & 35.5\% & 26.0\% & 58.4\% & 60.0\%  \\
    \bottomrule
  \end{tabular}
  }
\end{table}

\subsubsection{Baselines.}
\par We compare GPT4Rec with three types of baselines: (1) experience replay-based baselines, which includes Inverse Degree Sampling (Inverse)~\citep{ahrabian2021structure} and ContinualGNN~\citep{wang2020streaming}. (2) knowledge distillation-based baselines, which includes Topology-aware Weight Preserving (TWP)~\citep{liu2021overcoming}, GraphSAIL~\citep{xu2020graphsail}, SGCT~\citep{wang2021graph}, MGCT~\citep{wang2021graph} and LWC-KD~\citep{wang2021graph}. (3) parameter isolation-based baselines: DEGC~\citep{he2023dynamically}. (4) Vanilla Finetune baseline that initializes with parameters from the previous segment and fine-tunes using only the current segment's data.

\subsubsection{Reproducibility.} We apply grid search to find the optimal hyper-parameters for each model.  The ranges of hyper-parameters are $\{32, 64, 96, 128\}$ for size $L$ of node-level prompts $P$, the size $K$ of structure-level prompts $Q$ and the size $N$ of cross-view-level prompts $Z$. The range of disentangled view number is $\{2,4,8,16\}$. Adam optimizer~\citep{kingma2014adam}  is employed to minimize the training loss. Other parameters are tuned on the validation dataset and we save the checkpoint with the best validation performance as the final model. We use the same evaluation metrics \textbf{Recall@K} (abbreviated as \textbf{R@K}) and \textbf{NDCG@K} (abbreviated as \textbf{N@K}) following previous studies~\citep{xu2020graphsail,wang2021graph,he2023dynamically}.
All models are run five times with different random seeds and reported the average on a single NVIDIA GeForce RTX 3090 GPU.


\begin{table*}[t]
\huge
\centering
\caption{The average performance with MGCCF as our base model. The numbers in bold indicate statistically significant improvement (p \textless { .01}) by the pairwise t-test comparisons over the other baselines.}
\label{tab:overall}
\renewcommand{\arraystretch}{1.5}
\resizebox{\linewidth}{!}{
\begin{tabular}{c|l|cccccccccc|c|c}
\hline
         \multirow{2}*{Dataset} & \multirow{2}*{Metric} & (a) & (b) & (c) & (d) & (e) & (f) & (g) & (h) & (i) & (j) & (k)& (l)\\
         \cline{3-14}
         &  & Finetune & Uniform$^\star$ & Inverse$^\star$ & ContinualGNN$^*$ & TWP$^\dagger$ & GraphSAIL$^\dagger$ & SGCT$^\dagger$ & MGCT$^\ddagger$ & LWC-KD & DEGC & GPT4Rec & Improv.\\
         \hline
         \hline
    
\multirow{2}{*}{Taobao2014}     
& Recall@20        & 0.0412& 0.0308 & 0.0323  & 0.0311  & 0.0398 & 0.0395 & 0.0423     &  0.0421 & 0.0440 & \underline{0.1082}  &  \textbf{0.1127} & +4.16\%  \\\cline{2-14}
                          & NDCG@20        & 0.0052 & 0.0038 & 0.0039  & 0.0035  & 0.0050 & 0.0051   &  0.0054 & 0.0055 & 0.0059   &  \underline{0.0142}    &  \textbf{0.0149}   &  +4.93\%  \\
\midrule
\multirow{2}{*}{Taobao2015}     
& Recall@20        & 0.4256 & 0.4194 & 0.4217  & 0.4203  & 0.4320 & 0.4371 & 0.4411 & 0.4446   &  0.4512 & \underline{0.4892}    & \textbf{0.5018}   &  +2.58\%   \\\cline{2-14}
                          & NDCG@20        & 0.0117 & 0.0113  & 0.0114  & 0.0111  & 0.0122 & 0.0126   &  0.0129 & 0.0131 & 0.0139 & \underline{0.0167}     &  \textbf{0.0172}   & +2.99\%   \\
\midrule
\multirow{2}{*}{Netflix}   
& Recall@20        & 0.3359 & 0.3298 & 0.3321  & 0.3089  & 0.3428 & \underline{0.3470}   &  0.3300 & 0.3252 & 0.2616    & 0.2949 &\textbf{0.3508}   & +1.10\%  \\\cline{2-14}
                          & NDCG@20        & 0.0580 & 0.0533 & 0.0536   & 0.0491  & 0.0580 & \underline{0.0583}    &  0.0572 & 0.0562 & 0.0482 & 0.0520  &  \textbf{0.0589}   & +1.03\%   \\
\midrule
\multirow{2}{*}{Foursquare}     & Recall@20        & 0.1154 & 0.1024 & 0.1063  & 0.1056  & 0.1048 & 0.1086   &  0.1255 &0.1192 & 0.1280 &\underline{0.1425}    & \textbf{0.1477}   & +3.65\%    \\\cline{2-14}
                          & NDCG@20        & 0.0115 & 0.0101 & 0.0107   & 0.0098  & 0.0104 & 0.0110 &0.0137    &  0.0127 & 0.0144 & \underline{0.0178}  &  \textbf{0.0185}   & +3.93\%  \\
\bottomrule
\end{tabular}}
\end{table*}

\subsection{Comparison with Baseline Methods}
\label{sec:overall_performance}
In Table~\ref{tab:overall}, we
show the average performance of different methods on four datasets
while choosing MGCCF as the base GCN recommendation model.

The traditional Finetune method, for instance, inherits and fine-tunes parameters from previous data segments. While this can be effective for incremental updates, it often leads to catastrophic forgetting when new learning overshadows previously acquired knowledge. GPT4Rec circumvents this issue through its adaptive integration of new data, preserving historical context alongside new insights. Experience replay methods like Uniform Sampling and Inverse Degree Sampling sample historical data to combine with new information. However, they may not always strike the right balance between old and new data, potentially missing nuanced changes in user-item interactions. GPT4Rec's prompt-based strategy offers a more precise response to evolving data patterns, ensuring a seamless blend of historical and current user preferences. Knowledge distillation and experience replay techniques employed in ContinualGNN, TWP, GraphSAIL, SGCT, MGCT, and LWC-KD focus on pattern consolidation. These methods can be effective but may not fully capture the dynamic nature of user-item interactions in streaming scenarios. GPT4Rec's dual-prompt strategy, responsive to both node-level and structure-level changes, provides a more granular understanding of evolving preferences. DEGC, which models temporal preferences and performs historical graph convolution pruning and expanding, is adept at isolating long-term preferences. However, it might not be as nimble in adapting to quick short-term shifts. In contrast, GPT4Rec's flexible framework allows for real-time adaptation to both long-term and immediate changes, offering a comprehensive view of user preferences.

GPT4Rec distinguishes itself by its efficient integration of new knowledge via prompts, making it both streamlined and resource-efficient. Moreover, GPT4Rec excels in maintaining a critical balance between preserving historical data and adapting to emerging trends. This balance is essential in dynamic streaming environments where both historical continuity and responsiveness to new patterns are necessary for accurate recommendations. The model achieves a more precise understanding of user-item relationships through its sophisticated use of node-level and structure-level prompts. These prompts enable GPT4Rec to adjust its recommendations based on the context of each interaction, considering both the individual characteristics of nodes (such as specific user preferences and item attributes) and the overall structure of the graph (such as the connectivity and clustering of nodes). This contextual sensitivity ensures that the model not only captures but also effectively interprets the complex dynamics within the user-item graph, leading to more precise and relevant recommendations.


\begin{table}
\centering
  \caption{The average performance with NGCF as our base model. The numbers in bold indicate statistically significant improvement (p \textless { .01}) by the pairwise t-test comparisons over the other baselines.}
  \label{tab:ngcf}
  \begin{tabular}{c|c c|c c}
    \hline
    \multirow{2}*{Model} & \multicolumn{2}{c|}{\textit{Taobao2014}} & \multicolumn{2}{c}{\textit{Netflix}} \\\cline{2-5}
              ~ & R@20 & N@20 & R@20 & N@20 \\
    \hline
    Finetune & 0.0304 & 0.0040 & 0.3131 & 0.0541 \\\hline
    Uniform & 0.0340 & 0.0038 & \underline{0.3263} & 0.0525 \\\hline
    Inverse & 0.0347 & 0.0039 & 0.3256 & 0.0518 \\\hline
    ContinualGNN & 0.0338 & 0.0036 & 0.3047 & 0.0479 \\\hline
    TWP & 0.0358 & 0.0047 & 0.3159 & 0.0531 \\\hline
    GraphSAIL  & 0.0318 & 0.0042 & 0.3245 & \underline{0.0554} \\\hline
    SGCT & 0.0350 & 0.0046 & 0.3044 & 0.0533 \\\hline
    MGCT & 0.0346 & 0.0045 & 0.2957 & 0.0511 \\\hline
    LWC-KD & 0.0380 & 0.0050 & 0.2496 & 0.0454 \\\hline
    DEGC & \underline{0.0961} & \underline{0.0123} & 0.2713 & 0.0486\\\hline
    GPT4Rec & \textbf{0.0972} & \textbf{0.0129} &  \textbf{0.3341} & \textbf{0.0573} \\\hline
  \end{tabular}
\end{table}

\subsection{Method Robustness Analysis}
To evaluate the robustness of our method across various datasets and different GNN backbones, we use different GNN models as the backbones of GPT4Rec, and evaluate on the Netflix and Taobao2014 datasets. The results can be found in Tables~\ref{tab:ngcf} and~\ref{tab:lightgcn}. 

We observe that our method consistently outperform baselines on both datasets. Moreover, we find that our GPT4Rec can be combined with different GNN backbones properly. For example, on Taobao2014 dataset, when using NGCF as the backbone GNN, GPT4Rec outperforms the Finetune baseline by 6.71\% and 5.91\% in terms of Recall@20 and NDCG@20, respectively.  These results underscore the robustness and adaptability of GPT4Rec to different datasets and GNN backbones, indicating its potential for broad applicability in diverse recommendation scenarios.


The underlying reasons for these improvements and the robust nature of GPT4Rec can be attributed to several factors. Firstly, the model's unique prompt-based approach allows for a more context-aware adaptation to evolving user preferences and item characteristics. This approach ensures that the recommendations are not only accurate but also relevant to the current data landscape. Secondly, GPT4Rec's ability to dynamically integrate new information while preserving valuable historical insights helps maintain a balance that is crucial for the accuracy and relevance of recommendations in continually changing environments. Moreover, GPT4Rec's flexible framework adapts effectively to the inherent characteristics of different GCN models. Whether it's the complicated user-item interaction modeling in NGCF or the simplified yet efficient structure of LightGCN, GPT4Rec enhances these base models by effectively addressing their limitations and capitalizing on their strengths.

\begin{table}
\centering
  \caption{The average performance with LightGCN as our base model. The numbers in bold indicate statistically significant improvement (p \textless { .01}) by the pairwise t-test comparisons over the other baselines.}
  \label{tab:lightgcn}
  \begin{tabular}{c|c c|c c}
    \hline
    \multirow{2}*{Model} & \multicolumn{2}{c|}{\textit{Taobao2014}} & \multicolumn{2}{c}{\textit{Netflix}} \\\cline{2-5}
              ~ & R@20 & N@20 & R@20 & N@20 \\
    \hline
    Finetune & 0.0339 & 0.0040 & 0.3179 & 0.0537 \\\hline
    Uniform & 0.0377 & 0.0041 & \underline{0.3289} & 0.0533 \\\hline
    Inverse & 0.0386 & 0.0042 & 0.3275 & 0.0530 \\\hline
    ContinualGNN & 0.0382 & 0.0041 & 0.3035 & 0.0475 \\\hline
    TWP & 0.0338 & 0.0040 & 0.3204 & 0.0542 \\\hline
    GraphSAIL  & 0.0342 & 0.0042 & 0.3282 & \underline{0.0544} \\\hline
    SGCT & 0.0342 & 0.0043 & 0.3073 & 0.0519 \\\hline
    MGCT & 0.0357 & 0.0047 & 0.2983 & 0.0516 \\\hline
    LWC-KD & 0.0402 & 0.0053 & 0.2571 & 0.0461 \\\hline
    DEGC & \underline{0.0975} & \underline{0.0125} & 0.2776 & 0.0491\\\hline
    GPT4Rec & \textbf{0.0979} & \textbf{0.0131} &  \textbf{0.3348} & \textbf{0.0567} \\\hline
  \end{tabular}
\end{table}

\begin{figure}[h]
\centering
\includegraphics[width=0.40\textwidth]{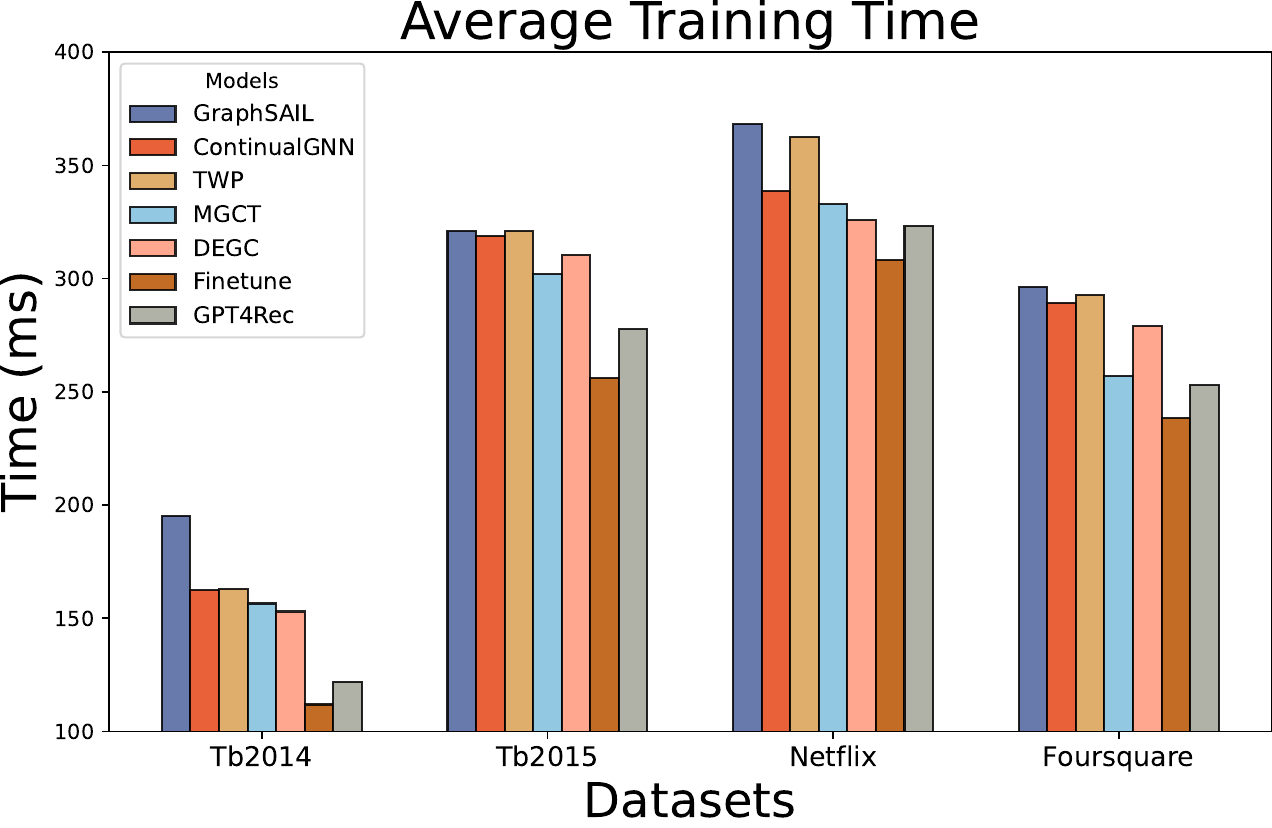}
\caption{Average training time per epoch}
\label{figure:efficiency}
\end{figure}

\subsection{Efficiency}
\par Figure \ref{figure:efficiency} illustrates the average training time per epoch of GPT4Rec in comparison with various baseline models. The results reveal that GPT4Rec not only matches the efficiency of the Finetune approach but also outperforms several other advanced models in terms of training speed. This high efficiency can be primarily attributed to the use of lightweight graph prompts. These prompts, despite their minimal computational demands, play a crucial role in seamlessly integrating new data into the model. By leveraging these compact yet powerful prompts, GPT4Rec avoids the need for extensive retraining or large-scale parameter adjustments that other models typically require. This streamlined approach ensures rapid adaptability to new information, significantly reducing computational overhead and training time, thereby enhancing overall efficiency.

\begin{figure}[t]
\centering
\subfigure[Recall@20]{
\includegraphics[width=0.45\textwidth]{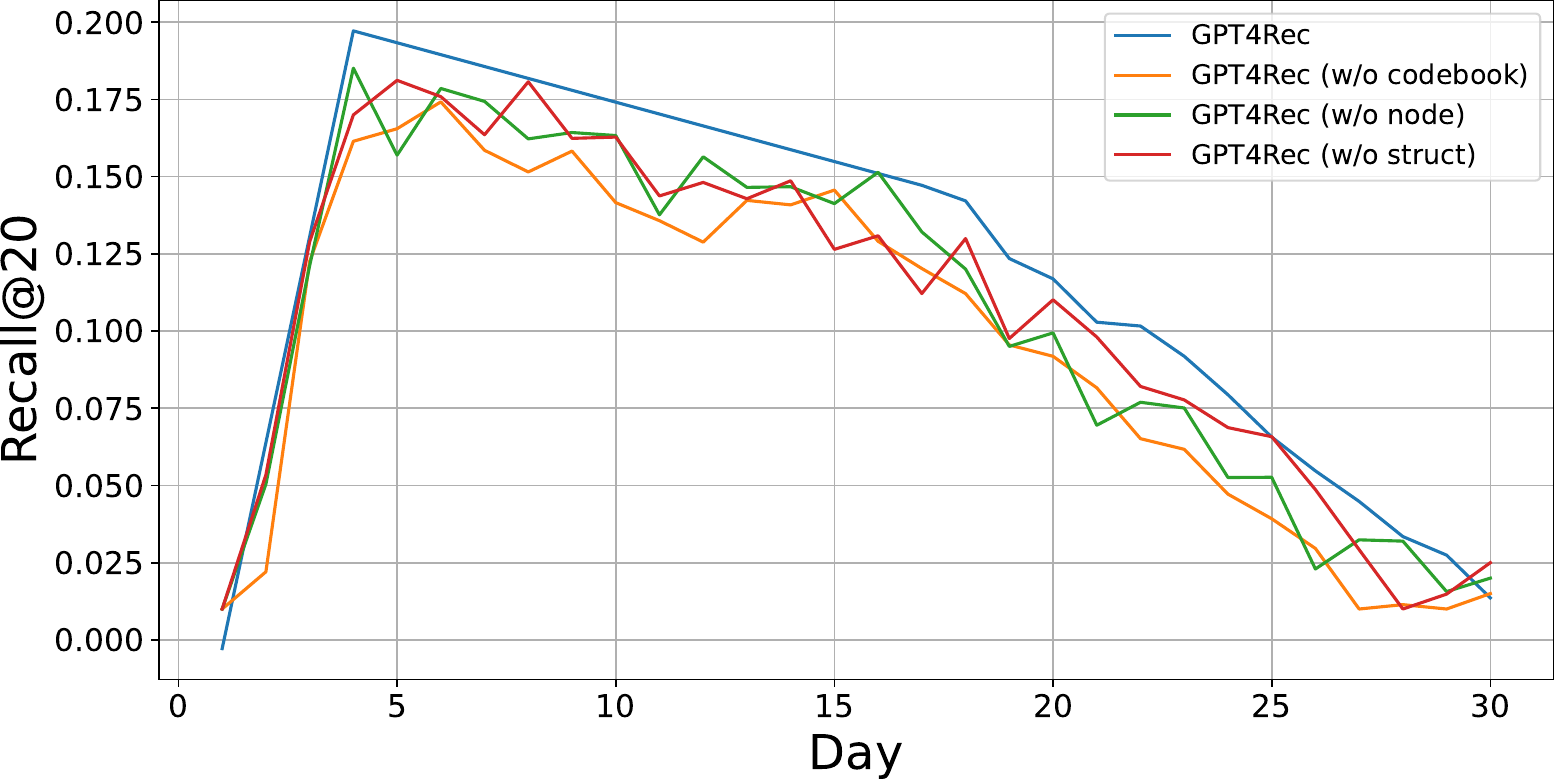}
\label{fig:recall}
} \\
\subfigure[NDCG@20]{
\includegraphics[width=0.45\textwidth]{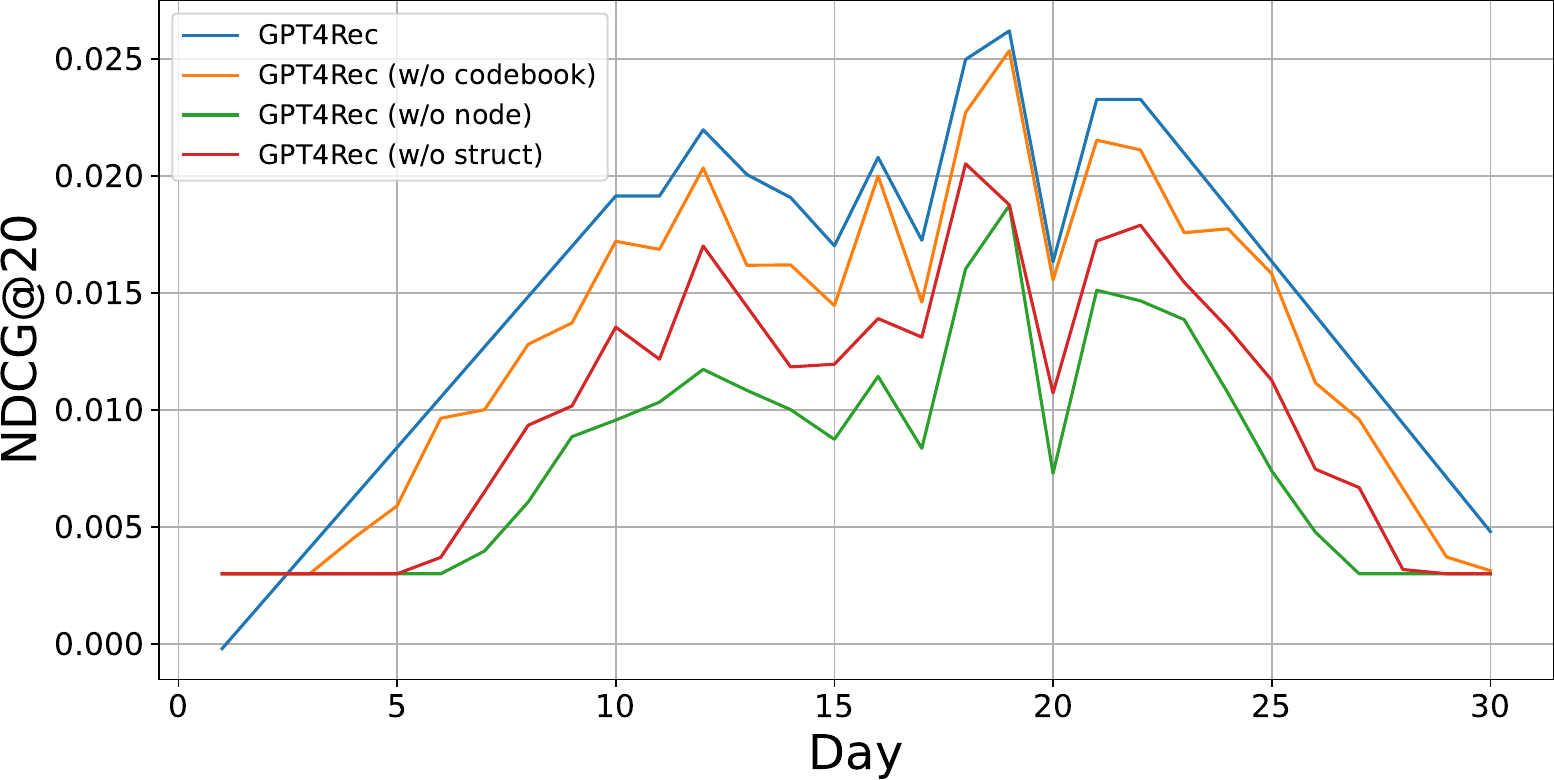}
\label{fig:ndcg}
}
\caption{Ablation study on Taobao2014 dataset.}
\label{fig:ablation}
\end{figure}

\subsection{Ablation Study}
In this section, we focus on GPT4Rec and test the efficacy of
its various designs in regard to the view disentangle, node-level prompt, structure-level prompt and view aggregation. The results are shown in Figure~\ref{fig:ablation}. We have the following observations:
\begin{itemize}[leftmargin=*,noitemsep,topsep=0pt]
    \item We evaluate the specific contribution of node-level prompts on adapting GPT4Rec to changes in individual nodes. By comparing versions of GPT4Rec with and without these prompts, we observe significant improvements in the model's response to shifts in user preferences and item attributes. These findings highlight the prompts’ role in providing context-specific adjustments to the model, enhancing its ability to personalize recommendations.
    \item The effectiveness of structure-level prompts is analyzed by assessing how well GPT4Rec adapts to overall structural changes in the graph. The comparison reveals that including structure-level prompts leads to a more accurate representation of the global interaction patterns, demonstrating their importance in capturing broader relationship dynamics in the graph.
    \item The view aggregation component is scrutinized. This aspect is crucial for reintegrating the disentangled views into a cohesive model output. The ablation study shows that effective aggregation is key to ensuring that the insights gained from the separate views are synergistically combined, leading to a more comprehensive understanding of the graph data.
    \item Finally, by comparing GPT4Rec's performance against baseline models that lack these features, it becomes evident that each component plays a vital role. The integration of view disentanglement, node-level prompts, structure-level prompts, and effective view aggregation sets GPT4Rec apart from its counterparts. This comprehensive design enables GPT4Rec to not only accurately capture the evolving nature of user-item interactions but also to do so in a way that is both efficient and scalable.
\end{itemize}

\subsection{Hyperparameter Study}
We conduct detailed hyperparameter studies on the hyperparameter of our model.

\begin{figure}[t]
\centering
\subfigure[Recall@20]{
\includegraphics[width=0.22\textwidth]{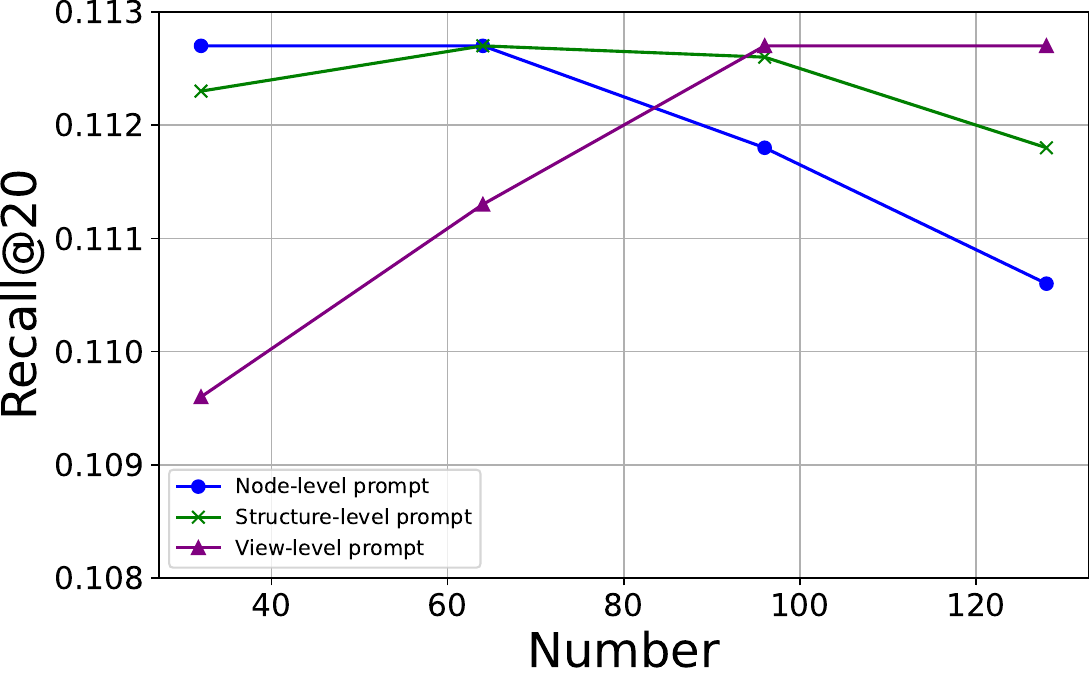}
\label{fig:recall}
} 
\subfigure[NDCG@20]{
\includegraphics[width=0.22\textwidth]{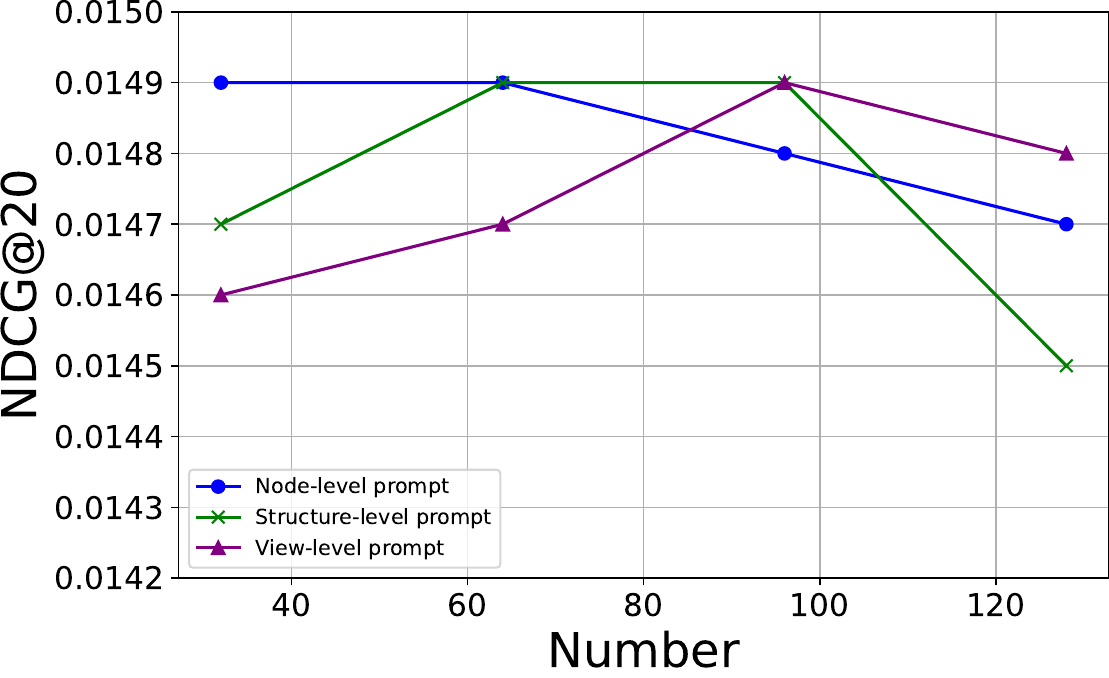}
\label{fig:ndcg}
}
\caption{Influence of Prompt Number on Taobao2014 dataset.}
\label{fig:prompt_size}
\end{figure}

\subsubsection{Prompt Size Study.} We first study the impact of varying the size of node-level prompts $P$ with size $L$, structure-level prompts $Q$ with size $K$, and cross-view-level prompts $Z$ with size $N$ on the Tb2014 dataset. The results are shown in Figure~\ref{fig:prompt_size}.

We find that initially increasing the size of these prompts generally enhances performance. Larger prompt sizes allow for a richer representation of complex user-item interactions and relationships, providing the model with a more diverse set of "hints" to interpret and integrate new information effectively. This improvement can be attributed to the model’s enhanced ability to capture and utilize the subtleties of user preferences and item characteristics. However, when the prompt size reaches a certain threshold, further increases only bring marginal benefits. Beyond this point, the additional capacity of the prompts does not significantly improve the model's interpretative power. Instead, the gains in performance become incremental and less impactful. Moreover, excessively large prompts introduce additional computational overhead, which may not be justifiable given the marginal gains in performance.

Interestingly, the optimal size of prompts also varies across different types: cross-view-level prompts require a larger size, while node-level prompts are most effective with a smaller size. This variation can be explained by considering the nature of changes each prompt type addresses. Cross-view-level prompts need to capture broader and more complex patterns of interaction across different graph views, which may necessitate a larger size for comprehensive representation. In contrast, node-level prompts, which target more specific and localized information, can achieve optimal performance with a smaller set of parameters.

\begin{figure}[t]
\centering
\includegraphics[width=0.3\textwidth]{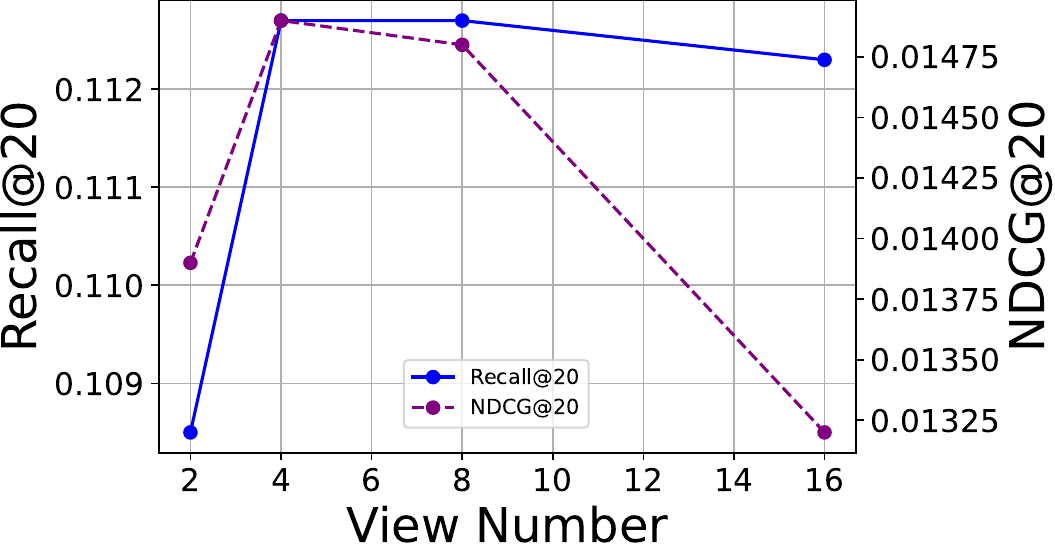}
\caption{Influence of View Number on Taobao2014 dataset.}
\label{fig:view_number}
\end{figure}

\subsubsection{View Size Study.} We explore the influence of the number of disentangled views on the Tb2014 dataset. The results are shown in Figure~\ref{fig:view_number}.  We find that the process of decomposing the graph into distinct views is a fundamental aspect of the GPT4Rec model. This disentanglement facilitates the model's understanding by allowing it to separately process diverse interaction dynamics between users and items. The results suggest that disentangling the graph into multiple views substantially improves the model’s capability to accurately represent and interpret the multifaceted relationship dynamics. Specifically, the disentangled views enable the model to address different aspects of the graph's structure in isolation, thus enhancing the quality of its recommendations by providing a more precise understanding of the interaction patterns.

In addition, an increase in the number of disentangled views generally corresponds to improved performance, primarily due to each view providing a distinct lens through which the model can perceive and process various aspects of user-item interactions. This multi-view approach enables a richer, more layered understanding of the data, as each view contributes unique insights into different facets of the graph, such as varying user preferences or item characteristics. 

However, our findings also highlight a critical threshold beyond which additional views may begin to hinder rather than help the model's performance. When the number of views crosses this threshold, the model encounters challenges in synthesizing and harmonizing these diverse perspectives. Additionally, managing an excessive number of views introduces significant computational challenges. The increased complexity not only escalates the computational costs but also amplifies the risk of the model overfitting. This issue underscores the importance of finding an optimal balance in the number of views, where the model can benefit from diverse perspectives without being overwhelmed by them or incurring prohibitive computational expenses.

\section{Conclusion}

In this paper, we propose GPT4Rec, a graph prompt tuning method for the continual learning in recommender systems. We propose to decouple the complex user-item interaction graphs into multiple semantic views, which enables the model to capture a wide range of interactions and preferences. The use of linear transformations in this decoupling process ensures that each view is distinctly represented while maintaining the overall structural integrity of the graph. The introduction of node-level, structure-level, and cross-view-level prompts in GPT4Rec is a significant methodological advancement. These prompts serve as dynamic and adaptive elements within the model, guiding the learning process and ensuring that the model remains responsive to new and evolving patterns within the graph. Extensive experiments validate our proposal.

\section*{Acknowledgements}
This work was supported by the Natural Science Foundation of China (No. 62372057, 62272200, 62172443, U22A2095).

\bibliographystyle{ACM-Reference-Format}
\balance
\bibliography{sample-base}


\end{document}